\documentclass[11pt,twoside]{amsart}
\usepackage[T1]{fontenc}
\usepackage[latin1]{inputenc}
\usepackage{amsmath,amssymb}

\usepackage[all]{xy}
\usepackage{color}
\usepackage{epsfig}
 \theoremstyle{plain}
\newtheorem{thm}{Theorem}[section]
  \theoremstyle{definition}
  \newtheorem{defn}[thm]{Definition}
 \theoremstyle{definition}
  \newtheorem{example}[thm]{Example}
\newtheorem{lemma}[thm]{Lemma}
\newtheorem{proposition}[thm]{Proposition}

\begin{document}

\title[Examples of non integer dimensional geometries]
{Examples of non integer dimensional geometries}


\author[R. Trinchero]{R. Trinchero}

\address{Instituto Balseiro y Centro At\'omico Bariloche.}  
\email{trincher@cab.cnea.gov.ar}
\begin{abstract}
Two examples of spectral triples with non-integer dimension
spectrum are considered. These triples involve commutative $C^{\star}$-algebras.
The first example has complex dimension spectrum and trivial differential
algebra. The other is a parameter dependent deformation of the canonical
spectral triple over $S^{1}$. Its dimension spectrum includes real non-integer values. It has a non-trivial differential algebra
and in contrast with the one dimensional case there are no junk forms
for a non-vanishing deformation parameter. The distance on this space
depends non-trivially on this parameter. 
\end{abstract}

\maketitle





\section{Introduction}

The dimension of a space is a basic concept of particular relevance
both in nature and in mathematics. Non-commutative geometry%
\footnote{The reader is referred to \cite{con} for an account of this theory.%
} provides a generalization of classical geometry. It is
the object of this work to study some examples of non-classical geometrical
objects within this framework. In particular, examples where the non-commutative
analog of dimension for the corresponding spaces can take non-integer
values are dealt with. The study of these geometries contributes
both to the implementation of physical ideas in these spaces and to
the exploration in concrete examples of the basic concepts in non-commutative geometry.

In the realm of quantum field theory, the widely employed dimensional
regularization technique\cite{C} provides a hint that such non-integer dimensional
spaces could be of relevance in physics. This study  intends to
go a step further in realizing those spaces as concrete geometrical
objects. In the field theoretic setting a desirable outcome in the
direction of the ideas mentioned above would be, for example, to obtain
an implementation of dimensional regularization at a Lagrangian level, or, in a more general setting,  
to really consider the possibility that physical space is non-integer
dimensional.

The general picture of non-commutative geometry shows that there are
many ways to construct spaces with non integer dimensions. In this
work we look for, so to say, the simplest and easiest to treat examples
of such spaces. A first step in this direction is to realize that
it is possible to get non-integer dimensional geometries even when the
underlying $C^{\star}$-algebra of the spectral triple is commutative. The examples presented
below are of this kind.

The generalization of the concept of dimension in non-commutative
geometry is provided by the Connes-Moscovici definition of dimension
spectrum\cite{CM}. Section 2 presents an example with complex dimension spectrum
but trivial differential algebra. It can be considered as a simple example
focused on getting a non-trivial dimension spectrum. Section
3 presents a deformation of the canonical spectral triple\footnote{The contents of section 2 and subsections 3.1 and 3.2 appeared first in reference \cite{fer}.} whose dimension spectrum includes real non-integer values. The main
results for this last example obtained in this work are related to
junk forms and distance. As mentioned in the abstract, no junk forms
appear for non-vanishing values of the deformation parameter $\alpha$. 
This fact leads to a differential algebra which involves forms of arbitrary 
degree and is neither anti-commutative nor commutative.
This situation may  be compared with the canonical case where junk forms play 
a crucial role in getting the usual differential algebra. This important  qualitative difference between the integer dimensional case and the deformed one is, in a certain vague sense, analog to the absence of poles in dimensionally regularized Feynman
integrals. 

Regarding distance, this turns out to be a function of the
coordinates and $\alpha$. This function involves a $\Gamma(\varepsilon)$($\varepsilon=1-\alpha$)
which does not come, as in dimensional regularization, from an analytical continuation of the 
factorial appearing in the solid angle integral. In fact, there is
no solid angle involved since we start from a support space for the commutative
$C^{\star}$-algebra in the spectral triple which is one dimensional.

\section{{Example with complex dimension spectrum.}}\label{ej1}

\subsection{General setting}

Here we deal with the triple $(\mathcal{A},\mathcal{H},\: D)$ where,

\subsubsection{Algebra}

$\mathcal{A}$ is the $\star$-algebra generated by $1,\, a,\,\bar{a}$
with the following multiplication rules,\begin{equation}
1\cdot a=a\cdot1=1\;\;,1\cdot\bar{a}=\bar{a}\cdot1=1\;\;,a\cdot\bar{a}=\bar{a}\cdot a\label{1}\end{equation}
where the star structure is given by,\begin{equation}
a^{\star}=\bar{a}\;\;.1^{\star}=1\label{2}\end{equation}
This algebra is isomorphic to the one of polynomials in a complex variable
$z$ and its complex conjugate $\bar{z}.$ The norm for this algebra
will be given below in terms of a representation in a Hilbert space
$\mathcal{H}$.

\subsubsection{Hilbert space}

$\mathcal{H}=L^{2}(S_{AP}^{1})=\mbox{square integrable anti-periodic functions on }$$S^{1}$
, thus{\footnote{This space can be identified with the one of square integrable sections of the associated spinor bundle corresponding to the anti-periodic spin structure over $S^1$. }}, \begin{equation}
\psi(x),\;0\leq x\leq2\pi\;\in\mathcal{H}\Rightarrow\psi(x)=-\psi(x+2\pi)\label{3}\end{equation}
with the scalar product given by,\begin{equation}
<\psi|\varphi>=\frac{1}{2\pi}\int_{0}^{2\pi}\;\psi^{\star}(x)\,\varphi(x)\label{4}\end{equation}
a basis for this Hilbert space is given by the functions,\begin{equation}
<x|k>=e^{i(k+1/2)x}\;\;,\, k\in\mathbb{Z}\label{5}\end{equation}
this basis is orthonormal,\begin{eqnarray}
<k|m>&=&\int_{x}\;<k|x><x|m>=\frac{1}{2\pi}\int_{0}^{2\pi}\; dx\; e^{-i(k+1/2)x}\, e^{i(m+1/2)x}\nonumber\\
&=&\delta_{km}\label{6}
\end{eqnarray}

\subsubsection{Representation of $\;\mathcal{A}$ in $\mathcal{H}$}

The following representation $\pi:\mathcal{A}\rightarrow\mathcal{H}$
is considered,\begin{equation}
\pi(a)=\sum_{k}|k+1/2|^{\varepsilon}|k><k|,\,\varepsilon\in\mathbb{C},\,|\varepsilon|<1,\,\Re e(\varepsilon)>0\label{7}\end{equation}
this definition has the property, \begin{equation}
\pi(a^{n})=\pi(a)^{n}\label{8}\end{equation}
the representant of the star element is defined by,

\begin{equation}
\pi(\bar{a})=\pi(a)^{\dagger}=\sum_{k}|k+1/2|^{\bar{\varepsilon}^{}}|k><k|\label{10}\end{equation}
where $\bar{\varepsilon}$ is the complex conjugate of $\varepsilon$. The
norm of the eigenvalues of $\pi(a)$ are always bounded as can be
seen from the following equation,

\begin{equation}
||k+1/2|^{-\varepsilon}|=e^{-\Re e(\varepsilon)ln|k+1/2|}\label{12}\end{equation}
and the requirement that $\Re e(\varepsilon)>0$. Thus the representative
of any element in the algebra is bounded. Defining the norm for the
algebra in terms of the operator norm in this representation leads
to a pre-$C^{\star}$-algebra structure in $\mathcal{A}$. Next this
pre-$C^{\star}$-algebra structure is completed to a $C\star$-algebra
structure in $\mathcal{A}$.

\subsubsection{The Dirac operator}

The Dirac operator is taken to be,\begin{equation}
D=-i\partial\label{13}\end{equation}
this operator is selfadjoint in $\mathcal{H}$, and its resolvent
is compact. Furthermore,\begin{equation}
[D,\pi(a)]=0\;\;\;,\forall a\in\mathcal{A}\label{14}\end{equation}
which is bounded. Thus the triple fulfills all the properties required
for it to be a spectral triple.

\subsection{Dimension spectrum of the triple}

The definition of dimension spectrum of a spectral triple is briefly
reviewed. Let $\delta$ denote the derivation $\delta:L({\mathcal{H}})\to L({\mathcal{H}})$
defined by, \begin{equation}
\delta(T)=[|D|,T]\qquad,T\in L({\mathcal{H}})\label{15}\end{equation}

Let $\mathcal{B}$ denote the algebra generated by the elements,\begin{equation}
\delta^{n}(\pi(a)),\; a\subset\mathcal{A},\;\; n\geq0\;(\delta^{0}(\pi(a))=\pi(a))\label{16}\end{equation}

\begin{defn}
{[}Connes-Moscovici]\label{cm} \emph{Discrete dimension spectrum.}
A spectral triple has discrete dimension spectrum $Sd$ if $Sd\subset\mathbb{C}$
is discrete and for any element $b\in\mathcal{B}$ the function, \begin{equation}
\zeta_{b}(z)=Tr[\pi(b)\:|D|^{-z}]\label{17}\end{equation}
 extends holomorphically to $\mathbb{C}/Sd$. 
\end{defn}
\noindent This definition is given in the paper \cite{CM}, where the
additional assumption is made that,\begin{equation}
\mathcal{A}\subset\cap_{n>0}Dom\delta^{n}\;\; and\;\;[D,\mathcal{A}]\subset\cap_{n>0}Dom\delta^{n}\label{18}\end{equation}
which clearly holds  for this case since $\mathcal{B}=Imag(\delta^{0})=\mathcal{A}$.
Thus the following generic $\zeta$-function is considered,

\begin{equation}
\zeta_{P(a,\bar{a})}(z)=Tr[\pi(P(a,\bar{a}))|-i\partial|^{-z}]\label{19}\end{equation}
where $P(a,\bar{a})$ is a polynomial in $a$ and $\bar{a}$ given
by,

\begin{equation}
P(a,\bar{a})=\sum_{n,m\in\mathbb{N}+0}\, p_{nm}\, a^{n}\:\bar{a}^{m}\label{20}\end{equation}
then,\begin{eqnarray}
\zeta_{P(a,\bar{a})}(z) & = & \sum_{k\in\mathbb{Z}}<\psi_{k}|P(|-i\,\partial|^{-\varepsilon},|-i\,\partial|^{-\bar{\varepsilon}})|\psi_{k}>|k+1/2|^{-z}\nonumber \\
 & = & 2\sum_{k\in\mathbb{\mathbb{N}}+0}P(|k+1/2|^{-\varepsilon},|k+1/2|^{-\bar{\varepsilon}})\;|k+1/2|^{-z}\nonumber \\
 & = & 2\sum_{n,m\in\mathbb{\mathbb{N}}+0}\, p_{nm}\:2^{z+n\varepsilon+m\bar{\varepsilon}}\sum_{k\in\mathbb{\mathbb{N}}+0}|2k+1|^{-n\varepsilon-m\bar{\varepsilon}-z}\nonumber \\
 & = & 2\sum_{n,m\in\mathbb{\mathbb{N}}+0}\, p_{nm}\:(2^{z+n\varepsilon+m\bar{\varepsilon}}-1)\;\zeta_{R}(z+n\varepsilon+m\bar{\varepsilon})\label{21}\end{eqnarray}
where $\zeta_{R}(z)$ denotes the Riemann $\zeta$-function, which has
a simple pole at $z=1$. Therefore the poles of $\zeta_{P(a,\bar{a})}(z)$
are located at,\begin{equation}
z=1-n\varepsilon-m\bar{\varepsilon}\;\;,n,m\in\mathbb{\mathbb{N}}+0\label{22}\end{equation}
Eq. \ref{22} gives the dimension
spectrum for this space. This spectrum is plotted  in Fig.1 for $\varepsilon= 1+i/2$.
\begin{figure}[!h]\label{fff1}
\begin{center}
\begin{picture}(0,0)%
\includegraphics{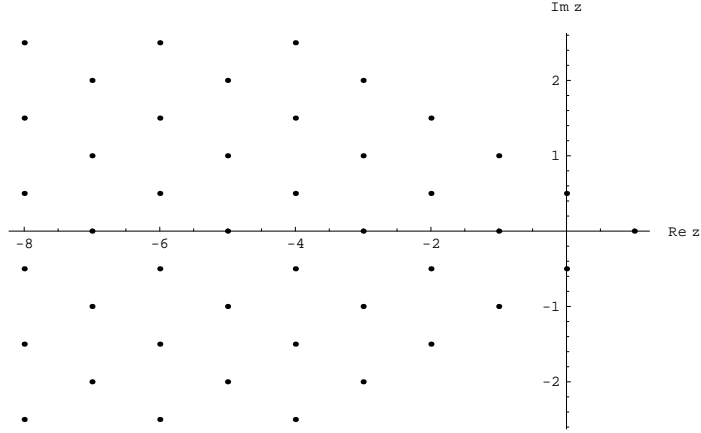}%
\end{picture}%
\setlength{\unitlength}{3947sp}%
\begingroup\makeatletter\ifx\SetFigFont\undefined%
\gdef\SetFigFont#1#2#3#4#5{%
  \reset@font\fontsize{#1}{#2pt}%
  \fontfamily{#3}\fontseries{#4}\fontshape{#5}%
  \selectfont}%
\fi\endgroup%
\begin{picture}(4410,2724)(2476,-3235)
\end{picture}
\caption{Discrete dimension spectrum for the triple of sec. \ref{ej1} with $\varepsilon= 1+i/2$.}
\label{ds}
\end{center}
\end{figure}

\section{Deformation of the canonical spectral triple on $S^{1}$}

\subsection{General setting}

Here we deal with the spectral triple $(\mathcal{A},\mathcal{H},\: D)$
where, 
\subsubsection{Algebra}
$\mathcal{A}=F(S^{1})$ is the commutative $C^*$-algebra
of smooth functions over $S^{1}$. This algebra  can be characterized by generators,
$b_{n}\;,n\in\mathbb{Z}$ satisfying the relations,\[
b_{n}b_{m}=b_{n+m}\;\;,b_{n}^{*}=b_{-n}\]
 
\subsubsection{Hilbert space}
$\mathcal{H}$=$L^{2}(S_{AP}^{1})$ is the Hilbert space of square integrable
anti-periodic functions over $S^{1}$, basis $|k>\;,k\in\mathbb{Z}$,
such that \begin{equation}
<x|k>=e^{i(k+1/2)x},n\in\mathbb{Z},x\in[0,2\pi]\:\; and\;\;<l|k>=\delta_{lk}\label{gs2}\end{equation}

\subsubsection{Representation of $\;\mathcal{A}$ in $\mathcal{H}$}
\begin{equation}
\pi(b_{n})=\sum_{k\in\mathbb{Z}}|k+n><k|\label{gs4}\end{equation}

\subsubsection{The Dirac operator}
$D$ $\in L(\mathcal{H})$ is given by, \begin{equation}
D=\sum_{k\in\mathbb{Z}}f_{\alpha}(k)\,|k><k|\label{gs3}\end{equation}
The following conditions on the functions $f_{\alpha}(k)$ are required
,

\begin{itemize}
\item $\lim_{\alpha\rightarrow0}f_{\alpha}(k)=k$. This condition ensures that
in this limit the canonical spectral triple is recovered. 
\item $\exists C>0 :|f_{\alpha}(k+n)-f_{\alpha}(k)| < C \;\; \forall k,n$.
This implies that $[D,\pi(k)]$ is bounded in $L(\mathcal{H})$. 
\item $f_{\alpha}(k)=f_{\alpha}(k)^{*}$ implies that $D$ is selfadjoint. 
\item $\lim_{k\rightarrow\infty}\frac{1}{f_{\alpha}(k)-\lambda}\rightarrow0,\qquad\forall\lambda\in\mathbb{R}$
implies that $D$ has Compact resolvent. 
\end{itemize}
The following choice of $f_{\alpha}(k)$will be considered from now
on,\begin{equation}
f_{\alpha}(k)=\frac{(k+1/2)^{\varepsilon}+[(k+1/2)^{\varepsilon}]^{*}}{2}\label{gs5}\end{equation}
 where $\varepsilon\in\mathbb{R},\;\;\varepsilon=1-\alpha,\;\alpha>0$
and the power is defined by the principal branch of the logarithm,
i.e.,\begin{equation}
(k+1/2)^{\varepsilon}=e^{\varepsilon\ln|k+1/2|+i\varepsilon\arg(k+1/2)}\label{gs6}\end{equation}
 where $\arg(k+1/2)=0\; or\;\pi$. This implies that the power in
(\ref{gs6}) is real only for integer $\varepsilon$. This fact explains why 
the complex conjugate is summed up  in (\ref{gs5}). Using
$(\ref{gs6})$ it follows that,\begin{equation}
f_{\alpha}(k)=\left\{ \begin{array}{c}
|k+1/2|^{\varepsilon}\;\;,k\ge0\\
|k+1/2|^{\varepsilon}\cos(\pi\varepsilon)\;\;,k<0\end{array}\right.\label{gs7}\end{equation}

This choice fulfills all the conditions required above.

\subsection{Dimension spectrum of the triple}\label{ssds}

The property (\ref{18}) holds for the choice (\ref{gs5}). According
to definition \ref{cm} the discrete dimension spectrum is the union
of the poles in $\zeta_{b}(z)$ for any $b\in\mathcal{B}$. Each of
such poles can be considered to describe the dimension of a certain
piece of the whole space. If we take elements $b$ in the image
of $\delta^{0}$, i.e. $b\in\mathcal{A}$, the following zeta function
is obtained, \begin{equation}
\zeta_{b_{n}}(z)=\delta_{n,0}\sum_{k}\;|f_{\alpha}(k)|^{-z}\label{ds14}\end{equation}
 For the choice (\ref{gs5}) the last equation leads to,\begin{eqnarray}
\zeta_{b_{n}}(z) & = & \delta_{n,0}(\sum_{k=0}^{\infty}\;|k+1/2|^{-\varepsilon z}+\sum_{k=-1}^{-\infty}|k+1/2|^{-\varepsilon z}\:\cos(\pi\varepsilon))\nonumber \\
 & = & \delta_{n,0}[1+\:\cos(\pi\varepsilon)]\;\;\sum_{k=0}^{\infty}\;|k+1/2|^{-\varepsilon z}\nonumber \\
\nonumber \\ & = & \delta_{n,0}[1+\:\cos(\pi\varepsilon))](2^{\varepsilon z}-1)\zeta_{R}(\varepsilon z)\label{ds15}\end{eqnarray}
 where $\zeta_{R}(z)$ denotes the Riemann zeta function . Thus $\zeta_{b_{n}}(z)$
has a simple isolated pole at,\begin{equation}
z=1/\varepsilon \label{ds16}\end{equation}
Since we are simply searching for a deformation of the canonical triple
with non-integer dimension we shall concentrate on this particular
piece of this \char`\"{}space\char`\"{}.

\subsection{Differential algebra}

Basic definition, \begin{equation}
\pi(db_{n})=[D,\pi(b_{n})]=\sum_{l}[f_{\alpha}(l+n)-f_{\alpha}(l)]\,|l+n><l|\label{da1}\end{equation}
 We look for junk forms, i.e., forms $\omega$ such that, \begin{equation}
\pi(\omega)=0\qquad,\pi(d\omega)\neq0\label{da2}\end{equation}

\begin{itemize}
\item For the case of $\omega$ a zero form, the only solution of $\pi(\omega)=0$
is $\omega=0$, and $\pi(d\omega)=0$ if $\omega=0$. Thus in this
case there are no junk 1-forms. 
\item For the case of $\omega$ a generic one form, \begin{eqnarray}
\pi(\omega) & = & \pi(\sum_{m,n}\alpha_{mn}b_{m}db_{n})\nonumber \\
&=&\sum_{m,n,l}\alpha_{mn}\,(f_{\alpha}(l+n)-f_{\alpha}(l))\;|l+m+n><l|\nonumber \\
 & = & \sum_{r,n,l}\alpha_{r-n,n}\,(f_{\alpha}(l+n)-f_{\alpha}(l))\;|l+r><l|\label{da3}\end{eqnarray}
 Now the operators $|r><s|$ and $|p><q|$ are linearly independent
if $r\neq p$ and/or $s\neq q$. Thus there can only be cancellations
in the r.h.s. of (\ref{da3}) between terms with $r=m+n$ and $l$
fixed . Hence the vanishing of $\pi(\omega)$ implies,\begin{equation}
\sum_{n}v_{n}^{(r)}\: w_{}^{(l)}(n)=0\;\;\;\;,\forall r,l\label{da7a}\end{equation}
 where,\begin{equation}
v_{n}^{(r)}=\alpha_{r-n,n}\;\;\;, w_{}^{(l)}(n)=f_{\alpha}(l+n)-f_{\alpha}(l)\label{da7b}\end{equation}
 Next, the differential of such 1-form $\omega$ with $\pi(\omega)=0$
is considered, \begin{eqnarray}
\pi(d\omega) & = & \sum_{m,n}\alpha_{mn}\,\pi(db_{m}\, db_{n})\nonumber \\
 & = & \sum_{m,n,k,l}\alpha_{mn}{}_{}\,w^{(k)}(m)\,w^{(l)}(n)|k+m><k|l+n><l|\nonumber \\
 & = & \sum_{m,n,l}\alpha_{mn}\,w^{(l+n)}(m)\,w^{(l)}(n)|l+m+n><l|\nonumber \\
 & = & \sum_{r,n,l}\alpha_{r-n,n}\,w^{(l+n)}(r-n)\,w^{(l)}(n)|l+r><l|\label{da7c'}\end{eqnarray}
\end{itemize}
below the issue of junk forms is studied both for the canonical and the deformed spectral triples.
The following  review of the canonical case will hopefully make clearer the difference with the deformed case.

\textbf{1.} $\mathbf{\alpha=0}$ $\;(f_{0}(k)=k+1/2)$. Eqs. (\ref{da7a})
and (\ref{da7b}) leads to, \begin{equation}
0=\pi(\omega)=\sum_{m,n,l}\alpha_{m,n}\, n\;|l+m+n><l|\label{da4}\end{equation}
 So, the equation to solve is, \begin{equation}
0=\sum_{n}v_{n}^{(r)}\; n=0\;\;,\forall r\label{da5}\end{equation}
 Two ways of dealing with this equation will be considered: 

\begin{enumerate}
\item [(i)] First we consider the equation, \begin{equation}
(\mathbf{v},\mathbf{N})=0\label{da6}\end{equation}
 where $\mathbf{N}$ is the vector $(-N,-N+1,\cdots,-1,0,1,\cdots,N)$
in a vector space $\mathbb{\mathbb{C}}^{2N+1}$ , $(,)$ is the canonical
scalar product in that space and the equation should be solved for
the vector $\mathbf{v}$, i.e. we look for all vectors orthogonal
to $\mathbf{N}$. Therefore, there are $2N$ independent solutions
to this equation, which can be chosen to be the following, 

\begin{enumerate}
\item For each $I=1,\cdots,N,\;\; v_{I}=1=v_{-I}$ the rest zero are $N$
solutions. 
\item $v_{0}=1$ the rest zero is solution. 
\item For $I=2,\cdots,N\;\;,v_{1}=I=-v_{-1}\;\;,v_{I}=-1=-v_{-I}$
the rest zero are $N-1$ solutions. 
\end{enumerate}
these are a set of $2N$ independent solutions, any other solution
can be expressed as a linear combination of these ones. Thus, the
solutions for $\alpha_{mn}$ are given for each $m+n=r$ fixed in
terms of the ones above by, \begin{equation}
\alpha_{r-n,n}=v_{n}^{(r)}\label{da7}\end{equation}
 For different values of $r$, there can correspond different solutions
for $v_{n}^{(r)}.$ 

\item [(ii)] Defining $v(x)$ by,\begin{equation}
v^{(r)}(x)=\sum_{n}v_{n}^{(r)}\, e^{i(n+1/2)x}\;\;,0\leq x\leq2\pi\label{da7c}\end{equation}
 eq.(\ref{da5}) can be written as,
\begin{equation}
-i\frac{d}{dx}\,\left(v^{(r)}(x)\, e^{-ix/2}\right)|_{x=0}=0\;\;\;\forall\, r\label{da7d}\end{equation}
which implies that $v^{(r)}(x)e^{-ix/2}$ can be any periodic function
with vanishing derivative at the origin. Note that the solutions appearing
in 1., in terms of $v(x)$ are: (a) $v(x)=2\, e^{ix/2}\,\cos kx\;\;,k=1,2,\cdots$
, (b)$e^{ix/2}$ (c) $v(x)=2i\, e^{ix/2}\,[k\sin(x)-\sin(kx)]\;\; k=2,3,\cdots$
whose linear combinations generate any function satisfying (\ref{da7d}).Now,
for the differential $\pi(d\omega)$ using (\ref{da7c}) we obtain,\begin{eqnarray}
\pi(d\omega) & = & \sum_{r,n,l}\alpha_{r-n,n}\,(r-n)\, n|l+r><l|\label{da7e}\end{eqnarray}
 which in terms of $v^{(r)}(x)$ can be written as, \begin{eqnarray}
\pi(d\omega)  =& & \sum_{r}\;\left\{ i\,\frac{d}{dx}\left[e^{-irx}-i\frac{d}{dx}(v^{(r)}(x)\, e^{-ix/2})^{}\right]\right\} _{x=0}\times \nonumber\\
&&\;\times \sum_{l}\,|l+r><l|\label{da7f}
\end{eqnarray}
the vanishing of the l.h.s. in (\ref{da7f}) imply the vanishing of
the expression between braces, which in turn implies, due to (\ref{da7d}),
that,\[
\left[\frac{d^{2}}{dx^{2}}(v^{(r)}(x)\, e^{-ix/2})\right]_{x=0}=0\]
thus since $\pi(\omega)=0$ does not imply $\pi(d\omega)=0$, there
are junk forms in this case. Indeed one can show that any two form
is a junk form in this case. %
\footnote{To make contact with other approaches to this problem, the following
1-form is considered,\begin{equation}
\omega_{m}=b_{m}\, db_{m}-db_{m}\, b_{m}\label{da7g}\end{equation}
 replacing (\ref{gs4}) for $\pi(b_{m})$ the following expression
is obtained,\begin{equation}
\pi(\omega_{m})=\sum_{l}\,[f_{\alpha}(l+m)-f_{\alpha}(l)-f_{\alpha}(l+2m)+f_{\alpha}(l+m)]\,|l+2m><l|\label{da7h}\end{equation}
which for the canonical case $f_{\alpha}(l)=l$ leads to $\pi(\omega_{m})=0$.
However,\begin{eqnarray}
\pi(d\omega_{m}) & = & \pi(db_{m})\pi(db_{m})\nonumber \\
 & = & \sum_{k,l}\,[f_{\alpha}(k+m)-f_{\alpha}(k)]\,[f_{\alpha}(l+m)-f_{\alpha}(l)]\,|k+m><k|l+m><l|\nonumber \\
 & = & \sum_{l}\,\,[f_{\alpha}(l+2m)-f_{\alpha}(l+m)]\,[f_{\alpha}(l+m)-f_{\alpha}(l)]\,|l+2m><l|\nonumber \\
 & = & \sum_{l}\,\, m^{2}\,|l+2m><l|\label{da7i}\end{eqnarray}
 which does not vanish. Furthermore, one can show that any two form $\omega$ can be written as,
\begin{equation}
 \omega= \sum_n \,f_n \,w_n
\end{equation}
where $f_n$ are elements of $\mathcal{A}$. Thus, leading to the result that in the canonical case any two form is a junk form.%
} 
\end{enumerate}
\textbf{2.} $\mathbf{\alpha\neq0}$. It is shown that, \begin{proposition}
There are no non-trivial junk 1-forms for $\alpha\neq0$. \end{proposition} 

\begin{proof}
Replacing (\ref{gs5}) in (\ref{da7b}) leads to,\begin{eqnarray}
w^{(l)}(n) & = & f_{\alpha}(l+n)-f_{\alpha}(l)\nonumber \\
 & = & \frac{(l+n+1/2)^{\varepsilon}+[(l+n+1/2)^{\varepsilon}]^{*}-(l+1/2)^{\varepsilon}-[(l+1/2)^{\varepsilon}]^{*}}{2}\label{da7j}\end{eqnarray}
 Employing the binomial expansion for the first two terms in the last
equation leads to,,\begin{eqnarray}
w_{}^{(l)}(n) & = & \sum_{j=1}^{\infty}\,\frac{1}{j!}\,\varepsilon(\varepsilon-1)\cdots(\varepsilon-j)\,\frac{1}{2}\left\{ [(l+1/2)^{\varepsilon-j}+[(l+1/2)^{\varepsilon-j}]^{*}\right\} \;\; n^{j}\nonumber \\
 & = & \sum_{j=1}^{\infty}\,\frac{1}{j!}\,\varepsilon(\varepsilon-1)\cdots(\varepsilon-j)\,|l+1/2|^{\varepsilon-j}\cos(\pi(\varepsilon-j))^{(1-sg(l))/2}\;\; n^{j}\label{da7k}\end{eqnarray}
 which is convergent for $|n/l|<1$ .

Next it is shown that, \begin{lemma} \begin{eqnarray}
\sum_{n}v_{n}^{(r)}\: w^{(l)}(n) & = & 0\;\;\;\;,\forall r,l\label{da81}\\
 & \Downarrow\nonumber \\
\sum_{n}v_{n}^{(r)}\, n^{m} & = & 0\;\; m=1,2,\cdots\label{da82}\end{eqnarray}
 \end{lemma} 
\begin{proof}
This follows from the following argument. Multiplying (\ref{da7k})
by $|l+1/2|^{1-\varepsilon}$ and taking the limit
$l\to\infty$, leaves only the first term in (\ref{da7k}) leading
to, \begin{equation}
\lim_{l\to\infty}|l+1/2|^{1-\varepsilon}w_{}^{(l)}(n)=n\label{da9}\end{equation}
 thus multiplying (\ref{da81}) by $|l+1/2|^{1-\varepsilon}$ and
taking the limit implies, \begin{equation}
\sum_{n}v_{n}^{(r)}\, n=0\label{da10}\end{equation}
 next multiplying (\ref{da81}) by $|l+1/2|^{2-\varepsilon}$, using
(\ref{da7k}) , (\ref{da10}) and taking the limit
$l\to\infty$ leads to, \begin{equation}
\sum_{n}v_{n}^{(r)}\, n^{2}=0\label{da11}\end{equation}
 continuing with this procedure leads to (\ref{da82}). 
\end{proof}
Replacement of the binomial expansion for $f_{\alpha}(n)$ in (\ref{da7c'})
and using (\ref{da82}) implies \begin{equation}
\pi(d\omega)=0\label{da7n}\end{equation}
 thus there are no non-trivial junk 1-forms for any of the two choices
of $f_{\alpha}(k)$. 
\end{proof}
Similar, but lengthier, arguments lead to the absence of junk forms
of any order.

The absence of junk forms has drastic consequences for the construction of the corresponding differential algebra.
No quotient is needed. In the canonical case this quotient leads  to the vanishing of any form of degree greater than one. In the deformed case there are forms of arbitrary degree and, as can be seen from expression (\ref{da3}), the differential algebra is neither anti-commutative nor commutative. This constitutes a remarkable situation since starting from a commutative $C^{\star}-$algebra in the triple, the deformed choice of the Dirac operator leads to a quite involved non-commutative differential algebra.

\subsection{Distance}

Let $\bar{\mathcal{A}}$ denote the closure of $\mathcal{A}$. A linear
functional $\varphi:\bar{\mathcal{A}}\to\mathbb{C}$ is called a state.
The set of all states being denoted by $S(\bar{\mathcal{A}})$. The
distance $d(\varphi,\psi)$ between two states $\varphi$ and $\psi$
is defined by,\begin{equation}
d(\varphi,\psi)=sup|\varphi(a)-\psi(a)|\label{d1}\end{equation}
 where the supremum is taken over all elements $a\in\bar{\mathcal{A}}$
such that,\begin{equation}
\parallel[D,a]\parallel\leq1\label{d2}\end{equation}
 where $\parallel\cdot\parallel$ denotes the operator norm. The l.h.s.
of (\ref{d2}) is given by,\begin{equation}
\parallel[D,a]\parallel=sup_{\psi}\frac{|[D,a]^{}|\psi>|}{||\psi>|}\label{d3}\end{equation}
The vector whose norm appears in the r.h.s. of (\ref{d3}) is given
by,\begin{eqnarray}
[D,a]^{}|\psi & > & =\sum_{l,m}\; a_{l-m}\,[f_{\alpha}(l)-f_{\alpha}(m)]\;\psi_{m}\;|l>\label{d4}\end{eqnarray}
where $\psi_{m}$ denotes the Fourier component of $|\psi>$($|\psi>=\sum_{m}\psi_{m}|m>$).
The fact that for $\alpha\neq0$ the quantity $f_{\alpha}(l)-f_{\alpha}(m)$
does not depend only on $l-m$ makes the evaluation of the norm of
this vector for arbitrary $|\psi>$ and $a(x)$ a non-trivial task.
It is true however that,\begin{equation}
f_{\alpha}(l)-f_{\alpha}(m)=f_{\alpha}(l-m)+O(\alpha)\label{d5}\end{equation}
therefore defining the following Liouville $\varepsilon$-derivative\footnote{This definition is an example of what is known in the literature as fractional derivative. See for example ref. \cite{old} for an overview of the subject.}(which
reduces to the usual derivative for $\varepsilon=1(\alpha=0)$ ),\begin{equation}
\frac{d^{\epsilon}}{dx}g(x)=i\;\sum_{l}\; g_{l}\, f_{\alpha}(l)e^{-i(l+1/2)x}\label{d6}\end{equation}
(\ref{d4}) implies,\begin{eqnarray}
<x|[D,a]^{}|\psi> & = & \sum_{l,m}\; a_{l-m}\, f(l-m)\;\psi_{m}\;<x|l>+O(\alpha)\nonumber \\
 & = & \psi(x)\,(-i\frac{d^{\epsilon}}{dx}a(x))+O(\alpha)\label{d7}\end{eqnarray}
thus,\begin{equation}
\parallel[D,a]\parallel=sup_{\psi}\sqrt{\frac{\int\, dx\:|\psi(x)|^{2}\,|\frac{d^{\epsilon}}{dx}a(x)|^{2}}{\int\, dx\:|\psi(x)|^{2}}}+O(\alpha)\label{d8''}\end{equation}
hence condition (\ref{d2}) implies,\begin{equation}
|\frac{d^{\epsilon}}{dx}a(x)|\leq1+O(\alpha)\;\;\forall x^{}\label{d9''}\end{equation}

\begin{example}
\emph{(i) The canonical triple\cite{con} over $S_{1}$.} 

In this case $f_{\alpha}(l)=l$ and the $\varepsilon$-derivative
reduces to the usual derivative. Thus (\ref{d2}) leads to, \begin{equation}
|\frac{d^{}}{dx}a(x)|\leq1\;\;\;\;\forall\, x\label{d8}\end{equation}
 This last inequality implies that the functions $a(x)$ over which
the supremun is taken are Lipschitz, i.e. , they satisfy%
\footnote{To see this note that (\ref{d8}) implies\begin{equation}
\int_{x_{0}}^{x_{1}}\, dx\,|\frac{d^{}}{dx}a(x)|\leq x_{1}-x_{0}\;\;\;\;\;,x_{1}>x_{0}\label{d81}\end{equation}
 noting that,\begin{equation}
\left|\int_{x_{0}}^{x_{1}}\, dx\,\frac{d^{}}{dx}a(x)\right|\leq\int_{x_{0}}^{x_{1}}\, dx\,|\frac{d^{}}{dx}a(x)|\label{d82}\end{equation}
 eq. (\ref{d80}) is obtained, i.e.,\begin{equation}
\left|\int_{x_{0}}^{x_{1}}\, dx\,\frac{d^{}}{dx}a(x)\right|=|a(x_{1})-a(x_{0})|\leq|x_{1}-x_{0}|\label{d83}\end{equation}
},\begin{equation}
|a(x_{1})-a(x_{0})|\leq|x_{1}-x_{0}|\label{d80}\end{equation}
 Next consider the state $\delta_{x_{0}}$ defined by,\begin{equation}
\delta_{x_{0}}(a(x))=a(x_{0})\label{d9}\end{equation}
 the distance between the two states $\delta_{x_{1}}$ and $\delta_{x_{0}}$
is, according to (\ref{d1}),\begin{equation}
d(\delta_{x_{1}},\delta_{x_{0}})=sup|\delta_{x_{1}}(a)-\delta_{x_{0}}(a)|=sup|a(x_{1})-a(x_{0})|\label{d10}\end{equation}
 thus (\ref{d80}) implies,\begin{equation}
d(\delta_{x_{1}},\delta_{x_{0}})\leq|x_{1}-x_{0}|\label{d101}\end{equation}
 On the other hand, consider the distance function to a fixed point $x_{0}$ ,\begin{equation}
d_{x_{0}}(x)=|x-x_{0}|\label{d84}\end{equation}
 which is clearly a function satisfying,\begin{equation}
\parallel[D,d_{x_{0}}]\parallel\leq1\label{d85}\end{equation}
 thus, it is clear that,\begin{equation}
d(\delta_{x_{1}},\delta_{x_{0}})\geq|d_{x_{0}}(x_{1})-d_{x_{0}}(x_{0})|=|x_{1}-x_{0}|\label{d111}\end{equation}
 hence, (\ref{d101}) and (\ref{d111}) imply that,\begin{equation}
d(\delta_{x_{1}},\delta_{x_{0}})=|x_{1}-x_{0}|\label{d112}\end{equation}

(ii) $\alpha\neq0$

In what follows we neglect the $O(\alpha)$ corrections in (\ref{d9''}). It turns out to be convenient to define the following kernel%
\footnote{The last equality in eq.(\ref{d113}) is obtained replacing in the
r.h.s. of the first equality the following integral representation
of the fractional power of $l+1/2-i\eta,$\[
(l+1/2-i\eta)^{-\varepsilon}=\frac{i^{\varepsilon}}{\Gamma(\varepsilon)}\int_{0}^{\infty}\, d\tau\,\tau^{\varepsilon-1}\, e^{-i(l+1/2-i\eta)\tau}\]
 Note the role played by the parameter $\eta$ for the convergence
of the integral in the r.h.s. of the last equation.%
},\begin{eqnarray}
d_{-\varepsilon}(x-y) & = & \lim_{\eta\to0}\;\sum_{l}\,\frac{1}{2\,\cos(\pi\varepsilon)}\,\left\{ [-(l+1/2-i\eta)]^{-\varepsilon}+c.c.\right\} \, e^{-i(l+1/2)(x-y)}\nonumber \\
 & = & \lim_{\eta\to0}\:\sum_{l}\,|l+1/2-i\eta|^{-\varepsilon}\,[\cos(\pi\varepsilon)]^{\frac{sg(l)-1}{2}}\, e^{-i(l+1/2)(x-y)}\nonumber \\
 & = & \frac{1}{2\,\cos(\pi\varepsilon)}\,\frac{|x|^{\varepsilon-1}}{\Gamma(\varepsilon)}\, e^{i\frac{\pi}{2}\varepsilon sg(x)}\label{d113}\end{eqnarray}
which has the following useful property, \[
\int_{-\pi}^{\pi}\, dy\, d_{-\varepsilon}(x-y)\;\frac{d^{\varepsilon}}{dy}a(y).=a(x)\]
and consider the following operator,\[
D_{-\varepsilon}^{\alpha\beta}g(x)=\frac{1}{2\pi}\int_{-\pi}^{\pi}\, dy\:|d_{-\varepsilon}(\beta-y)-d_{-\varepsilon}(\alpha-y)|\, g(y)\]
which is chosen to produce the following result,\begin{eqnarray*}
D_{-\varepsilon}^{\alpha\beta}|\frac{d^{\varepsilon}}{dx}a(x)| & = & \frac{1}{2\pi}\int_{-\pi}^{\pi}\, dy\:|d_{-\varepsilon}(\beta-y)-d_{-\varepsilon}(\alpha-y)\,\frac{d^{\varepsilon}}{dy}a(y)|\\
 & \geq & |\frac{1}{2\pi}\int_{-\pi}^{\pi}\, dy\:(d_{-\varepsilon}(\beta-y)-d_{-\varepsilon}(\alpha-y))\,\frac{d^{\varepsilon}}{dy}a(y)|\\
 & = & |a(\beta)-a(\alpha)|\end{eqnarray*}
therefore applying the operator $D_{-\varepsilon}^{\alpha\beta}$to
both sides of (\ref{d9''}) implies,\[
|a(\beta)-a(\alpha)|\leq D_{-\varepsilon}^{\alpha\beta}1\]
it is not difficult to see that (\ref{d113}) implies that $D_{-\varepsilon}^{\alpha\beta}1$
depends only of $|\beta-\alpha|$, i.e.,\[
D_{-\varepsilon}^{\alpha\beta}1=\frac{D_{-\varepsilon}(|\beta-\alpha|)}{2\pi}\]
Using the last equality in (\ref{d113}) leads to the following expression
for $D_{-\varepsilon}(x)$ in,\[
D_{-\varepsilon}(x)=\mathcal{F}(\varepsilon)\left\{ C(\varepsilon)|x|^{\varepsilon}+\frac{1}{\varepsilon}[(\pi+x)^{\varepsilon}+(\pi-x)^{\varepsilon}-2\pi^{\varepsilon}]\right\} \]
where,\begin{eqnarray*}
\mathcal{F}(\varepsilon) & = & |2\Gamma(\varepsilon)\cos(\pi\varepsilon)|^{-1}\\
C(\varepsilon) & = & \int_{0}^{1}dy\;\left\{ |1-y|^{\varepsilon-1}+|y|^{\varepsilon-1}-2|1-y|^{\varepsilon-1}\,|y|^{\varepsilon-1}\;\cos(\pi\varepsilon)\right\} \end{eqnarray*}
for the case $\varepsilon=1$ these expressions give $\mathcal{F}(1)=1/2$
and $C(1)=2$ thus leading to the usual result, $D_{-1}(x)=|x|$.

To obtain the distance function one repeats the argument between eqs.
(\ref{d80}) and (\ref{d112}) obtaining,\[
d(\delta_{x_{1}},\delta_{x_{0}})=\frac{D_{-\varepsilon}(|x_{1}-x_{0}|)}{2\pi}\]
In Fig. 2,3 and 4 the function $D_{-\varepsilon}(x)$ is plotted for $\varepsilon=1,0.9$ and $0.8$ respectively. According to subsection \ref{ssds}
these values of $\varepsilon$ correspond to spaces with dimension $d=1/\varepsilon=1,1.111\cdots , 1.25$. Including the $O(\alpha)$ corrections to (\ref{d9''}) would produce corrections of $O(\alpha^2)$ to $D_{-\varepsilon}(x)$.
\end{example}
\section{Outlook}
The results presented in this work may be considered as first steps in investigating these geometries.
Many additional issues deserve to be studied. Some of them are: integration, operatorial structure of the differential algebra, metric volume element, gauge theories ....
\section{Acknowledgments}
I want to express my gratitude to S. Capriotti, C. Fosco, S. Grillo and F. Mendez for valuable comments and suggestions. 
\newpage

\begin{figure}[htbp]
\begin{center}
\begin{picture}(0,0)%
\includegraphics{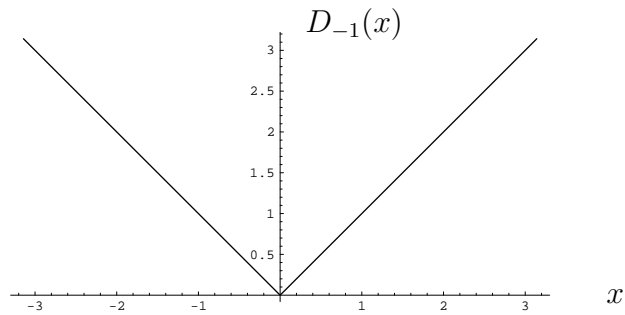}%
\end{picture}%
\setlength{\unitlength}{4144sp}%
\begingroup\makeatletter\ifx\SetFigFont\undefined%
\gdef\SetFigFont#1#2#3#4#5{%
  \reset@font\fontsize{#1}{#2pt}%
  \fontfamily{#3}\fontseries{#4}\fontshape{#5}%
  \selectfont}%
\fi\endgroup%
\begin{picture}(3660,1862)(1801,-1797)
\put(3646,-106){\makebox(0,0)[lb]{\smash{{\SetFigFont{12}{14.4}{\rmdefault}{\mddefault}{\updefault}{\color[rgb]{0,0,0}$D_{-1}(x)$}%
}}}}
\put(5446,-1726){\makebox(0,0)[lb]{\smash{{\SetFigFont{12}{14.4}{\rmdefault}{\mddefault}{\updefault}{\color[rgb]{0,0,0}$x$}%
}}}}
\end{picture}
\caption{Distance for a space of dimension $d=1$.}
\label{dis1}
\end{center}
\end{figure}
\begin{figure}[htbp]
\begin{center}
\begin{picture}(0,0)%
\includegraphics{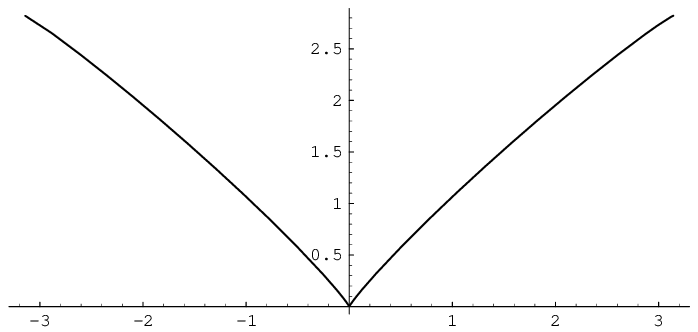}%
\end{picture}%
\setlength{\unitlength}{4144sp}%
\begingroup\makeatletter\ifx\SetFigFont\undefined%
\gdef\SetFigFont#1#2#3#4#5{%
  \reset@font\fontsize{#1}{#2pt}%
  \fontfamily{#3}\fontseries{#4}\fontshape{#5}%
  \selectfont}%
\fi\endgroup%
\begin{picture}(3570,1744)(2116,-2399)
\put(3916,-826){\makebox(0,0)[lb]{\smash{{\SetFigFont{12}{14.4}{\rmdefault}{\mddefault}{\updefault}{\color[rgb]{0,0,0}$D_{-0.9}(x)$}%
}}}}
\put(5671,-2311){\makebox(0,0)[lb]{\smash{{\SetFigFont{12}{14.4}{\rmdefault}{\mddefault}{\updefault}{\color[rgb]{0,0,0}$x$}%
}}}}
\end{picture}
\caption{Distance for a space of dimension $d=1.11\cdots$.}
\label{dis09}
\end{center}
\end{figure}
\begin{figure}[!h]
\begin{center}
\begin{picture}(0,0)%
\includegraphics{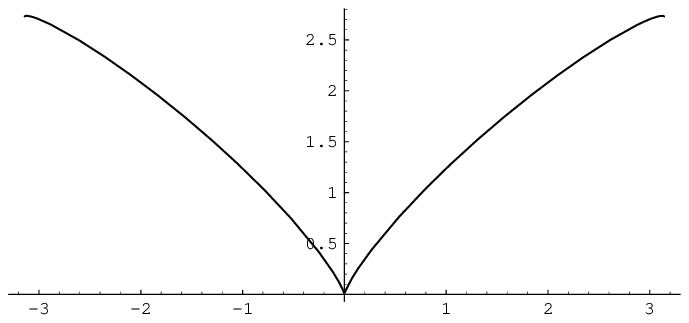}%
\end{picture}%
\setlength{\unitlength}{4144sp}%
\begingroup\makeatletter\ifx\SetFigFont\undefined%
\gdef\SetFigFont#1#2#3#4#5{%
  \reset@font\fontsize{#1}{#2pt}%
  \fontfamily{#3}\fontseries{#4}\fontshape{#5}%
  \selectfont}%
\fi\endgroup%
\begin{picture}(3480,1675)(1801,-1520)
\put(3511,-16){\makebox(0,0)[lb]{\smash{{\SetFigFont{12}{14.4}{\rmdefault}{\mddefault}{\updefault}{\color[rgb]{0,0,0}$D_{-0.8}(x)$}%
}}}}
\put(5266,-1456){\makebox(0,0)[lb]{\smash{{\SetFigFont{12}{14.4}{\rmdefault}{\mddefault}{\updefault}{\color[rgb]{0,0,0}$x$}%
}}}}
\end{picture}%
\caption{Distance for a space of dimension $d=1.25$.}
\label{dis08}
\end{center}
\end{figure}
\bibliographystyle{my-h-elsevier}

\end{document}